\documentclass{article}
\usepackage{amsmath,amsfonts,amsthm}
\usepackage{graphicx}
\newtheorem{definition}{Definition}
\newtheorem{lemma}{Lemma}
\newtheorem{proposition}{Proposition}
\newtheorem{theorem}{Theorem}
\theoremstyle{remark}
\newtheorem{remark}{Remark}

\title{Optimal importance sampling for L\'evy Processes}
\author{Adrien Genin \\ Opus Finance Research \\ and \\Laboratoire de Probabilit\'es
et Mod\`eles Al\'eatoires\\ Universit\'e Paris Diderot --- Paris 7 \and Peter Tankov\\ Laboratoire de Probabilit\'es
et Mod\`eles Al\'eatoires\\ Universit\'e Paris Diderot --- Paris 7}

\date{}
\begin{document}

\maketitle

\begin{abstract}
We develop generic and efficient importance sampling estimators for Monte Carlo evaluation of prices of single- and multi-asset European and path-dependent options in asset price models driven by L\'evy processes, extending earlier works which focused on the Black-Scholes and continuous stochastic volatility models \cite{guasoni2008optimal,robertson2010sample}. Using recent results from the theory of large deviations on the path space for processes with independent increments \cite{leonard2000large}, we compute an explicit asymptotic approximation for the variance of the pay-off under an Esscher-style change of measure. Minimizing this asymptotic variance using convex duality, we then obtain an easy to compute asymptotically efficient importance sampling estimator of the option price. Numerical tests for European baskets and for Asian options in the variance gamma model show consistent variance reduction with a very small computational overhead. 
\end{abstract}
 


\section{Introduction}

The aim of this paper is to develop efficient and easy to implement importance sampling estimators of expectations of functionals of L\'evy processes, corresponding to option prices in exponential L\'evy models. 

L\'evy processes are stochastic processes with stationary independent
increments. They are used as models for asset
prices when jump risk is important, either directly (as in the variance gamma model \cite{madan98}, normal inverse Gaussian process \cite{bns_nig}, CGMY model \cite{finestructure}) or as building blocks for other models (affine processes, stochastic volatility models with jumps, local L\'evy models \cite{cgmy.local.levy} etc.). To model a financial market with a L\'evy process, we assume that the market consists of a risk-free asset $S^{0}_t\equiv 1$
and $n$ risky assets $S^{1},\dots,S^{n}$ such that $$S^i_t = S^i_0
e^{X^i_t},$$ where $(X^1,\dots,X^n)$ is a L\'evy process under the risk-neutral probability $\mathbb P$.

We consider a derivative written on $(S^i)_{1\leq i \leq n}$ with
pay-off $P(S)$ which depends of the entire trajectory of
the stocks. We are interested in computing the price of this derivative, given by the risk-neutral expectation $\mathbb E \left[ P(S) \right]$. 

Several methods for computing this expectation are available in the literature. When the price process $S$ is one-dimensional and the pay-off $P$ only depends on the terminal value $S_T$, the Fourier method of Carr and Madan \cite{carrmadan} may be used.  When the dimension of $S$ is low (say, up to 3-4), and the pay-off is only weakly path-dependent, such as for barrier or American options, one can use deterministic numerical methods for partial
  integro-differential equations \cite{voltchkova05}, Fourier
  time stepping  \cite{lord2008fast,fang2009pricing} and related deterministic
  methods. Finally, for high dimensional problems, or in the case of strong path dependence, the Monte Carlo method, on which we focus in this paper is the only one available.

The standard Monte Carlo estimator of $\mathbb E \left[ P(S) \right]$ is defined as the empirical mean
$$
\widehat P_N := \frac{1}{N} \sum_{j=1}^N P(S^{(j)}),
$$
where $S^{(j)}, j=1,\dots,N$ are i.i.d.~samples with the same law as $S$. 
Note that simulation methods exist for all parametric L\'evy
          models, including multidimensional L\'evy processes (see chapter 6 of \cite{levybook} for a review). 


The precision of standard Monte Carlo is often too low for
real-time applications, particularly when $\mathbb E[P(S)]$ is small compared to $\sqrt{\text{Var}\, P(S)}$ and various error reduction techniques must be applied.
\begin{itemize}
	\item The Multilevel Monte Carlo method (see \cite{giles2008multilevel} for a general introduction and \cite{xia2012multilevel,ferreiro2014multilevel} for an application to L\'evy models) reduces the variance by optimizing the number of discretization steps for each path. In practice, a large number of paths are simulated with a coarse discretization, and only a small number of paths are discretized finely. 
\item The Quasi Monte Carlo method (see \cite{leobacher2006stratified,avramidis2006efficient} for the case of jump diffusions / L\'evy processes) replaces the i.i.d.~samples of $S$ with well-chosen deterministic samples (low-discrepancy sequences). 
\item Finally, the importance sampling method, which is the focus of this paper, consists in simulating the paths of $S$ under a different probability measure, which allows a better exploration of the region of interest. See e.g.,  \cite{glasserman1999asymptotically} for an application in the context of Gaussian vectors, \cite{guasoni2008optimal} for the case of path-dependent options in the Black-Scholes model, \cite{robertson2010sample} for an application to stochastic volatility models, and \cite{kassim2015importance,kawai2006importance,kawai2008optimal} for applications to L\'evy processes / jump-diffusions.
\end{itemize}

The importance sampling estimator is based on the following identity, valid for any probability measure $\mathbb Q$, with respect to which $\mathbb P$ is absolutely continuous. 
$$
\mathbb E[P(S)] = \mathbb E^{\mathbb Q}\left[\frac{d\mathbb
    P}{d\mathbb Q}P(S)\right].
$$
This allows one to define the {importance sampling estimator}
$$
\widehat P^{\mathbb Q}_N := \frac{1}{N} \sum_{j=1}^N \left[\frac{d\mathbb
  P}{d\mathbb Q}\right]^{(j)} P(S^{(j)}_{\mathbb Q}),
$$
where $S^{(j)}_{\mathbb Q}$ are i.i.d.~sample trajectories of $S$ under the
measure $\mathbb Q$. For efficient variance reduction, one needs then to find a probability measure $\mathbb Q$
such that $S$ is {easy to simulate} under $\mathbb Q$ and the variance 
$$
\text{Var}_{\mathbb Q} \left[ P(S)\frac{d\mathbb P}{d\mathbb Q} \right] 
$$ 
is considerably smaller than the original variance $\text{Var}_{\mathbb P} \left[P(S) \right]$. 
 
For L\'evy processes, a natural choice of probability measure for importance sampling is given by the {Esscher transform}
$$
\frac{d\mathbb P^\theta}{d\mathbb P}\Big|_{\mathcal F_T} = \frac{e^{\langle \theta, X_T\rangle}}{\mathbb E\left[e^{\langle \theta, X_T\rangle}\right]},
$$
which is well defined for all $\theta \in \mathbb R^n$ such that 
$\mathbb E\left[e^{\langle \theta, X_T\rangle}\right]$. This choice was studied, e.g., in \cite{kawai2008optimal,kawai2006importance}, where the optimal variance reduction parameter $\theta^*$ was estimated numerically.

In this paper, to allow for more freedom in choosing the importance sampling probability for path-dependent payoffs, we propose to use the path-dependent Esscher transform,
$$
\frac{d\mathbb P^\theta}{d\mathbb P} = \frac{e^{\int_{[0,T]} X_t
    \cdot \theta(dt)}}{\mathbb E\left[e^{\int_{[0,T]} X_t
    \cdot \theta(dt)}\right]},
$$
where $\theta$ is a (deterministic) bounded $\mathbb R^n$-valued
signed measure on $[0,T]$. Under $\mathbb P^\theta$, the process $X$ has independent increments and is thus easy to simulate. 
The optimal choice of $\theta$ should minimize the variance of the estimator under $\mathbb P^\theta$,
$$
\text{Var}_{\mathbb P^\theta} \left( P \frac{d\mathbb P}{d\mathbb P^\theta}\right) = \mathbb E_{\mathbb P}\left[P^2\frac{d\mathbb P}{d\mathbb P^\theta}\right] - \mathbb E\left[P\right]^2.
$$

Importance sampling is most effective in the context of \emph{rare event simulation}, e.g., when $P(S)=0$ for most of the trajectories of $S$ under the original measure $\mathbb P$. Since the theory of large deviations is concerned with the study of probabilities of rare events, it is natural to use measure changes appearing in or inspired by the large deviations theory for importance sampling. We refer, e.g., to \cite{dupuis2004importance} and references therein for a review of this approach and to the above quoted references \cite{glasserman1999asymptotically,guasoni2008optimal,robertson2010sample} for specific applications to financial models.  

The main contribution of this paper, inspired by the work of Guasoni and Robertson \cite{guasoni2008optimal} in the setting of the Black-Scholes model, is to use the large deviations theory to construct an easily computable approximation to the optimal importance sampling measure $\theta^*_{opt}$. Namely, we use Varadhan's lemma and the pathwise large deviation principle for L\'evy processes due to Leonard \cite{leonard2000large} to derive a proxy for the variance of the importance sampling estimator which is much easier to compute than the true variance. We propose then to use the parameter $\theta^*$, obtained by minimizing this proxy, in the importance sampling estimator. Numerical illustrations in Section \ref{num.sec} show that the variance obtained by using $\theta^*$ instead of $\theta^*_{opt}$ is very close to the optimal one, and that a considerable variance reduction is obtained in all examples with very little computational overhead. 

When the logarithm of the pay-off $P$ is concave, which is the case in many applications, the proxy for the variance may be further simplified using convex duality. The computation of the asymptotically optimal parameter $\theta^*$ then reduces to one finite-dimensional optimization problem for European options and to the solution of one ODE system (Euler-Lagrange equations)  for the path-dependent ones. In other words, additional complexity is the same as in the case of the Black-Scholes model studied in \cite{guasoni2008optimal}, even though our model is much more general and complex. 

The rest of this paper is structured as follows. In Section \ref{ld.sec} we recall the notation and results from the theory of large deviations which are used in the paper. Section \ref{main.sec} provides a representation for the proxy of the variance, a simplified representation in the case of concave log-payoffs and an easy to verify criterion for concavity. Section \ref{ex.sec} presents explicit computations for European basket and Asian options. Numerical illustrations of these examples, in the context of the variance gamma model, are provided in Section \ref{num.sec}. Lastly, the Appendix contains a technical lemma.

\section{Pathwise large deviations for L\'evy processes}\label{ld.sec}
In this section we recall the known results on large deviations for L\'evy processes, which will be used in the sequel, and introduce all the necessary notation. We first formulate the large deviations principle (LDP) on abstract spaces. Let $\mathcal X$ be a Haussdorf topological space endowed with its Borel $\sigma$-field. A {rate function} is a $[0,\infty]$-valued lower semi-continuous function on $\mathcal X$. It is said to be a {good rate function} is its level sets are compact. 
\begin{definition}[Large Deviation Principle]
A family $\{X^\varepsilon\}$ of $\mathcal X$-valued random variables is said to obey a {LDP} in $\mathcal X$ with rate function $I$ if for each open subset $G \subset \mathcal X$ and each closed subset $F \subset \mathcal X$
$$
\lim\sup_{\varepsilon \to 0} \varepsilon \log\mathbb P\left[X^\varepsilon \in F\right] \leq - \inf_{x\in F}I(x)
$$ 
and
$$
\lim\inf_{\varepsilon \to 0} \varepsilon \log\mathbb P\left[X^\varepsilon \in G\right] \geq - \inf_{x\in F}I(x).
$$
\end{definition}

The following result is the famous Varadhan's lemma, which allows to evaluate limits of functions of $X^\varepsilon$ in the large deviations asymptotics. More precisely, we need an extension of this lemma, given in \cite{guasoni2008optimal}, which allows the function $\phi$ to take value $-\infty$. This is necessary since the pay-off of the option may take zero value on part of the domain, and the function $\phi$ will  contain the log-pay-off in the sequel. 

\begin{lemma}[Varadhan's lemma]
Suppose that $\{X_\varepsilon\}$ satisfies the LDP with a good rate
function $I : \mathcal X \to [-\infty, \infty[$ and let $\phi : \mathcal X \to
[-\infty,\infty[$ be such that the set $\{\phi>-\infty\}$ is open and $\phi$
is continuous on this set. Assume further that for some $\gamma >1$,
$$
\lim\sup_{\varepsilon\to 0} \varepsilon \log \mathbb E \left[e^{\frac{\gamma\phi (Z_\varepsilon)}{\varepsilon}}\right] < \infty.
$$
Then,
$$
\lim_{\varepsilon\to 0}\varepsilon\log\mathbb E \left[e^{\frac{\phi(Z_\varepsilon)}{\varepsilon}}\right] = \sup_{x\in \mathcal X} \left\{ \phi(x) - I(x)\right\}.
$$
\end{lemma}

%

We shall next recall the pathwise large deviation principle for L\'evy processes, but first, following \cite{leonard2000large}, we need to introduce specific topological spaces well suited for this application, and recall some preliminary results on L\'evy processes.

\paragraph{Spaces and topologies}
As usual, $D$ denotes the space of right-continuous functions with left limits (RCLL) on the interval $[0,T]$. The subspace of $D$ containing all functions on $[0,T]$ with bounded variation will be denoted by $V_r$. The symbol $M$ will denote the class of bounded $\mathbb R^n$-valued
measures on $[0,T]$. Note that there is a one-to-one correspondence between elements of $V_r$ and those of $M$: in particular, for every $\mu\in M$, the function $t\mapsto \mu([0,t])$ belongs to $V_r$. Let $\sigma(D,M)$ denote the topology on $D$ defined by 
$$
\lim_n y_n = y\ \Leftrightarrow\ \forall \mu \in M,\ \lim_n \int_{[0,T]} y_n d\mu = \int_{[0,T]} y d\mu.
$$
This topology is stronger than the topology of pointwise
  convergence but weaker than the uniform topology. 

For future reference, we let $V^{ac}_r$ denote the subspace of $V_r$ consisting of absolutely continuous functions $x$ such that $x_0=0$ and $\int_0^T |\dot x_s| ds<\infty$, equipped with the norm $\|x\| = \int_0^T |\dot x_s| ds$. 

\paragraph{Preliminaries on L\'evy processes}
Recall that the law of an $\mathbb R^n$-valued L\'evy process $X$ is characterized by its {L\'evy triplet}
{$(A,\nu,\gamma)$} via the L\'evy-Khintchine formula
$$
\mathbb E[e^{iuX_t}] = \exp t \left(i\langle u, {\gamma}\rangle - \frac{\langle{A} u,u\rangle}{2}
  + \int_{\mathbb R^n} (e^{i\langle u,x\rangle}-1-i\langle u,x\rangle 1_{|x|\leq 1}) {\nu}(dx)\right)
$$ 
Here, the L\'evy measure {$\nu$} is a positive measure on $\mathbb R^n$, which satisfies $$\int (|x|^2 \wedge 1) \nu(dx) <\infty$$ and governs the intensity of
jumps, the matrix $A$ is a positive definite $n\times n$ matrix, which corresponds to the covariance matrix of the diffusion component and the vector $\mu\in \mathbb R^n$ is related to the deterministic linear component of $X$. 
We shall need the following lemma, which is a direct consequence of Propositions 2.2 and B.1 in \cite{leonard2000large}. 
\begin{lemma}\label{Gfunc.lm}
Let $\theta \in M$ and let $X$ be a L\'evy process. Then,
$$
\log \mathbb E\left[e^{\int_{[0,T] } X_t \cdot \theta(dt)}\right] = \int_0^T G(\theta([t,T])) dt \ \in [0,\infty],
$$
where
$$
G(\lambda) = \log \mathbb E\left[ e^{\langle\lambda, X_1\rangle}\right].
$$

\end{lemma}
In the sequel, we shall make use of the following two assumptions on the L\'evy process $X$. 
\begin{itemize}
\item[(A1)] There exists $\lambda_0>0$ with $\int_{|x|>1}e^{\lambda_0 |x|} \nu(dx)<\infty$.
\item[(A2)] The function $G$ is lower semicontinuous and its effective domain is bounded. 
\end{itemize}

\paragraph{A large deviations principle for L\'evy processes}
In the following, we let $X$ denote a $\mathbb R^n$-valued L\'evy process on $[0,T]$ with L\'evy measure $\nu$ and without diffusion part $(A=0)$.
We introduce a family of L\'evy processes $(X^\varepsilon)_{\varepsilon >0}$ defined by {$X^\varepsilon_t = \varepsilon X_{t/\varepsilon}$}. The following theorem can be found in \cite{leonard2000large}. 
\begin{theorem}\label{leonard.thm}
Suppose that Assumption (A1) holds true. Then the family $\{X^\varepsilon\}$ satisfies the LDP in $D$ for the $\sigma(D,M)$-topology with the good rate function $\bar J(y)$ where
\begin{equation*}
\bar J(x) =	
\left\{ 
\begin{aligned}
&\sup_{\mu\in M}\left\{\int_{[0,T]} x_t \mu(dt) - \int_0^T
  G(\mu([t,T])) dt\right\} &\text{if } x \in V_r \\
&+\infty &\text{otherwise.}
\end{aligned}
\right.
\end{equation*}
\end{theorem}
Note that  De Acosta \cite{de1994large} proves an LDP for the uniform topology under the assumption that {all exponential moments are finite}.  However, this assumption is too strong in practice, since most financial models are based on L\'evy processes with exponential tail decay. 

The rate function appearing in Theorem \ref{leonard.thm} admits a more explicit expression (see section 6 in \cite{leonard2000large}). 
Define the Fenchel conjugate of $G$:
$$
L_a(v) = \sup_{\lambda\in\mathbb R^d} \left\{ \langle\lambda, v\rangle - G(\lambda) \right\}
$$
and its recession function
$$
L_s(v) = \lim_{u\to\infty} \frac{L_a(uv)}{u}.
$$
Then, 
\begin{equation*}
\bar J(x) =	
\left\{ 
\begin{aligned}
&\int_{[0,T]} L_a(\frac{d\dot x_a}{dt}(t))dt + \int_{[0,T]} L_s(\frac{d\dot x_s}{d\mu}(t))d\mu &\text{if } x \in V_r \\
&+\infty &\text{otherwise,}
\end{aligned}
\right.
\end{equation*}
where $\dot x = \dot x_a + \dot x_s$ is the decomposition of the measure $\dot x\in M$ in absolutely continuous and singular pars with respect to $dt$ and $\mu$ in any nonnegative measure on $[0,T]$, with respect to which $\dot x_s$ in absolutely continuous.

\section{Main results}\label{main.sec}

As mentioned in the introduction, our importance sampling estimator is based on the path-dependent Esscher transform,
$$
\frac{d\mathbb P^\theta}{d\mathbb P} = \frac{e^{\int_{[0,T]} X_t
    \cdot \theta(dt)}}{\mathbb E\left[e^{\int_{[0,T]} X_t
    \cdot \theta(dt)}\right]}
$$
where $\theta$ is a (deterministic) bounded $\mathbb R^n$-valued
signed measure on $[0,T]$.

The optimal choice of $\theta$ should minimize the variance of the estimator under $\mathbb P^\theta$,
$$
\text{Var}_{\mathbb P^\theta} \left( P \frac{d\mathbb P}{d\mathbb P^\theta}\right) = \mathbb E_{\mathbb P}\left[P^2\frac{d\mathbb P}{d\mathbb P^\theta}\right] - \mathbb E\left[P\right]^2
$$
Denote $H(X) = \log P(S)$. Then, using Lemma \ref{Gfunc.lm}, the minimization problem writes
$$
\inf_{\theta \in M}\mathbb E_{\mathbb P}\left[ \exp\left\{ 2H(X) - \int_{[0,T]} X_t
    \cdot \theta(dt) + \int_0^T G(\theta([t,T]))dt \right\}\right],
$$
where 
$$
G(\theta) = \langle \theta, {\mu}\rangle +  \int_{\mathbb R^n} (e^{\langle \theta,x\rangle}-1-\langle \theta,x\rangle 1_{|x|\leq 1}) {\nu}(dx).
$$

Given the possibly complex form of the log-payoff $H$, the above expression for the variance is difficult to minimize. Our approach is instead to minimize a proxy of the variance, which has a more tractable form.
Our first main result provides an expression for such a proxy, which we aim to minimize to obtain an asymptotically optimal variance reduction. 
\begin{proposition}\label{varadhanapplies}
Let Assumption (A1) hold true, and suppose that the set $\{x\in D: H(x)>-\infty\}$ is open and that $H$ is continuous on this set  for the $\sigma(D,M)$-topology and satisfies
$$
H(x) \leq A + B\sup_{s \in [0,T]} \sum_{i=1}^n |x^i_s|
$$
with $B< {\lambda_0}/{4n}$. 
Then, for every $\theta \in M$ such that $$\max_{0\leq t\leq T} |\theta([t,T])| < \lambda_0 - 4nB,$$ 
it holds that
\begin{align*}
&\lim_{\varepsilon \to 0}\varepsilon \log \mathbb E \left[
  e^{\frac{2H(X^\varepsilon) - \int_{[0,T]}  X^\varepsilon_t \cdot \theta(dt) }{\varepsilon}} \right]  = 
\sup_{x \in D} \left\{2H(x) - \int_{[0,T]}  x_t \cdot \theta(dt)  - \bar J(x ) \right\}. 
\end{align*}
\end{proposition}
\begin{proof}
Since the pay-off $H$ is assumed to be continuous, the continuity of the mapping
$$
x\mapsto 2H(x) - \int_{[0,T]}  x_t \cdot \theta(dt)
$$
for the $\sigma(D,M)$-topology follows from the definition of this topology. It remains to check the integrability condition of Varadhan's lemma.
By assumptions of the Proposition, we may choose $p>1$ and $q>1$ with $\frac{1}{p} + \frac{1}{q}=1$, as well as $\gamma>1$, such that 
 \begin{align*}
&q\max_{0\leq t\leq T}|\theta([t,T])| <\lambda_0 \quad \text{and}\quad 4 pn B <\lambda_0. 
\end{align*} 
Moreover, there exists $b>0$ with 
$$
\mathbb E\left[ (X^i_{t}-bt)e^{4B\gamma pn (X^i_{t}-bt)}\right]<0\quad \text{and}\quad \mathbb E\left[ (-X^i_{t}-bt)e^{4B\gamma pn (-X^i_{t}-bt)}\right]<0
$$
for all $t>0$. 
Then, by the assumption on $H$, the Cauchy-Schwarz inequality, and Lemma \ref{wienerhopf.lm}, the following estimates hold true:
\begin{align*}
&\lim\sup_{\varepsilon\to 0} \varepsilon \log \mathbb E \left[e^{\frac{\gamma(2H(X^\varepsilon) - \int_{[0,T]}  X^\varepsilon_t \cdot \theta(dt) ) }{\varepsilon}}\right] \\
& \leq 2A\gamma + \lim\sup_{\varepsilon\to 0} \varepsilon \log \mathbb E \left[e^{\frac{\gamma(2B\sup_{s\in [0,T]} \sum_{i=1}^n|X^{\varepsilon,i}_s| - \int_{[0,T]}  X^\varepsilon_t \cdot \theta(dt) ) }{\varepsilon}}\right]\\
& = 2A\gamma +\lim\sup_{\varepsilon\to 0} \varepsilon \log \mathbb E \left[e^{{\gamma(2B\sup_{s\in [0,T]}\sum_{i=1}^n |X^i_{s/\varepsilon}| - \int_{[0,T]}  X_{t/\varepsilon} \cdot \theta(dt) ) }}\right]\\
&\leq 2A\gamma + 2bnT+ \sum_{i=1}^n \lim\sup_{\varepsilon\to 0} \varepsilon \log \mathbb E\left[e^{4B\gamma pn\sup_{s\in [0,T]} (X^i_{s/\varepsilon}-bs/\varepsilon)}\right] \\&+\sum_{i=1}^n\lim\sup_{\varepsilon\to 0} \varepsilon \log \mathbb E\left[e^{4B\gamma pn\sup_{s\in [0,T]} (-X^i_{s/\varepsilon}-bs/\varepsilon)}\right]\\& + \lim\sup_{\varepsilon\to 0} \varepsilon \log \mathbb E \left[e^{- q\gamma\int_{[0,T]}  X_{t/\varepsilon} \cdot \theta(dt)  }\right]\\
& \leq  2A\gamma + 2bnT+ \sum_{i=1}^n \lim\sup_{\varepsilon\to 0} \varepsilon \log \mathbb E\left[e^{4B\gamma pn\sup_{s\geq 0} (X^i_{s}-bs)}\right] \\
& + \sum_{i=1}^n\lim\sup_{\varepsilon\to 0} \varepsilon \log \mathbb E\left[e^{4B\gamma pn\sup_{s\geq 0} (-X^i_{s}-bs)}\right] + \int_{[0,T]} G(-q\gamma\theta([t,T]))< \infty. 
\end{align*}

\end{proof}

The result of Proposition \ref{varadhanapplies} leads us to introduce the following definition. 
\begin{definition}\label{optvar.def}
We say that the variance reduction parameter $\theta^*$ is
asymptotically optimal if it minimizes 
$$
\sup_{x \in V_r} \left\{2H(x) - \int_{[0,T]}  x_t \cdot \theta(dt)  + \int_{[0,T]} G(\theta([t,T]))dt- \bar J(x ) \right\}
$$
over $\theta \in M$. 
\end{definition}


%

The optimization functional in the Definition \ref{optvar.def} is difficult to compute in practice, since the rate function $\bar J$ is usually not known explicitly. The following theorem shows that for concave log-payoffs, the computation of the optimal parameter $\theta^*$ is greatly simplified. European basket put options and many paht-dependent put-like payoffs encountered in practice are indeed concave.

\begin{theorem}\label{maindual.thm}
Let $H$ be concave and upper semicontinuous on $V^{ac}_r$ and assume that for every $x\in V_r$ there is a sequence $\{x_n\} \subset V^{ac}_r$ converging to $x$ in the $\sigma(D,M)$-topology and such that $H(x_n) \to H(x)$. Let Assumption (A2) be satisfied.  
Then,
\begin{multline}
\inf_{\theta \in M}\sup_{x \in V_r} \left\{2H(x) - \int_{[0,T]}  x_t \cdot \theta(dt) + \int_{[0,T]} G(\theta([t,T]))dt - \bar J(x ) \right\}  \\= 2\inf_{\theta \in M}\{\widehat H(\theta) +
\int_{[0,T]} G(\theta([t,T]))dt\}\label{dualtheta}
\end{multline}
where 
$$
\widehat H(\theta) = \sup_{x \in V_r}\{H(x) - \int_{[0,T]} x_t \theta (dt)\}.
$$
Moreover, if the infimum in the left-hand side of \eqref{dualtheta} is attained by $\theta^*$ then the same value $\theta^*$ attains the infimum in the right-hand side of \eqref{dualtheta}. 
\end{theorem}

\begin{remark}
The assumption that the effective domain of $G$ is bounded is the most restrictive. It is satisfied by models where the tail decay is exactly exponential, such as variance gamma, normal inverse gaussian, CGMY and their multidimensional versions. However, it rules out models with faster than exponential tail decay such as the celebrated Merton's model. We expect that for such models a similar result may still be shown, but one would need to use different, and slightly more complex methods (Orlicz spaces instead of $L^\infty$). To keep the length of the proof reasonable, we have chosen to present the argument in the case of a bounded domain. 
\end{remark}
\begin{proof}
Step 1. By assumption of the proposition, 
\begin{align}
&{\sup_{x\in V_r}\inf_{\mu\in M}}\{2H(x) - \int_{[0,T]} x_t\cdot \mu(dt)+ \int_{[0,T]}
  G(\mu([t,T])) dt\} \notag\\&\qquad= {\sup_{x\in V^{ac}_r}\inf_{\mu\in M}}\{2H(x) - \int_{[0,T]} x_t\cdot \mu(dt)+ \int_{[0,T]}
  G(\mu([t,T])) dt\}\notag\\
& \qquad=  {\sup_{x\in V^{ac}_r}\inf_{\mu\in M}}\{2H(x) - \int_{[0,T]} \dot x_t\cdot \mu([t,T])dt + \int_{[0,T]}
  G(\mu([t,T])) dt\}\notag\\
&  \qquad=  {\sup_{x\in V^{ac}_r}\inf_{y\in L^\infty([0,T])}}\{2H(x) - \int_{[0,T]} \dot x_t\cdot y_t dt + \int_{[0,T]}
  G(y_t) dt\},\label{step1minimax}
\end{align}
where the last equality follows by approximating an $L^{\infty}$ function with a sequence of continuous functions with bounded variations, and using the dominated convergence theorem.

Step 2. Our aim in this step is to show that $\sup$ and $\inf$ in \eqref{step1minimax} may be exchanged. We adapt the classical argument, which may be found, e.g., in \cite[paragraph 6.2]{ekeland1976convex}. Fix $\varepsilon>0$ and define, for $x\in V^{ac}_r$ and $y \in L^\infty([0,T])$, 
$$
L_\varepsilon(x,y) := 2H(x) - \varepsilon \|x\|^2 - \int_{[0,T]} \dot x_t \cdot y_t dt+ \int_{[0,T]}
  G(y_t) dt, 
$$
and $L(x,y):= L_0(x,y)$. 
For $x\in V^{ac}_r$, let 
$$
f_\varepsilon(x) = \inf_{y \in L^\infty([0,T])}L_\varepsilon(x,y). 
$$
The $\inf$ is attained by $\hat y(x)$ such that for each $t$, $\hat y_t(x)$
is the minimizer of $y\mapsto - \dot x_t \cdot y+  G(y)$. The function $f_\varepsilon$ is concave, upper semicontinuous and satisfies $f_\varepsilon(x)\to -\infty$ whenever $\|x\|\to \infty$.  Therefore, it attains its upper bound at the point $\hat x^\varepsilon$. Let $x\in V^{ac}_r$ and $\lambda\in(0,1)$. By the optimality of $\hat x^\varepsilon$ and $\hat y(\hat x^\varepsilon)$ and the concavity of $L_\varepsilon$ with respect to the first argument,  
\begin{align*}
f_\varepsilon(\hat x^\varepsilon) &\geq f_\varepsilon((1-\lambda)\hat x^\varepsilon + \lambda x)  = L_\varepsilon((1-\lambda)\hat x^\varepsilon + \lambda x,\hat y((1-\lambda)\hat x^\varepsilon + \lambda x))\\  &\geq (1-\lambda)L_\varepsilon(\hat x^\varepsilon,\hat y((1-\lambda)\hat x^\varepsilon + \lambda x)) + \lambda L_\varepsilon(x,\hat y((1-\lambda)\hat x^\varepsilon + \lambda x))
\\  &\geq (1-\lambda)L_\varepsilon(\hat x^\varepsilon,\hat y(\hat x^\varepsilon)) + \lambda L_\varepsilon(x,\hat y((1-\lambda)\hat x^\varepsilon + \lambda x)),
\end{align*}
which implies that 
$$
f_\varepsilon(\hat x^\varepsilon) \geq L_\varepsilon(x,\hat y((1-\lambda)\hat x^\varepsilon + \lambda x)). 
$$
Since the effective domain of $G$ is bounded, the function
$$
x\mapsto \sup_y\{x\cdot y - G(y)\}
$$
is Lipschitz continuous, and therefore, as $\lambda \to 0$, 
$$
L_\varepsilon(x,\hat y((1-\lambda)\hat x^\varepsilon + \lambda x)) \to L_\varepsilon(x,\hat y(\hat x^\varepsilon)),
$$
which implies that for  all $x\in V^{ac}_r$,
$$
f_\varepsilon(\hat x^\varepsilon) \geq L_\varepsilon(x,\hat y(\hat x^\varepsilon)),
$$
and therefore, for all $x\in V^{ac}_r$ and $y\in L^\infty([0,T])$,
\begin{align}
L_\varepsilon(x,\hat y(\hat x^\varepsilon)) \leq L_\varepsilon(\hat x^\varepsilon,\hat y(\hat x^\varepsilon)) \leq L_\varepsilon(\hat x^\varepsilon,y).\label{saddlepoint}
\end{align}
Since the family $\hat y(\hat x^\varepsilon)$ is bounded in $L^\infty$, by Banach-Alaoglu theorem there exists a sequence $\{\varepsilon_j\} $ and a point $\bar y$ such that $\hat y(\hat x^{\varepsilon_j})\to \bar y$ in the weak$^*$ topology of $L^\infty$ and hence in the weak topology of $L^p$ for all $p>1$.  Since
$$
y\mapsto \int_{[0,T]} G(y_t) dt
$$
is convex and lower semicontinuous in the strong topology of $L^p$ (see \cite{rockafellar68}), it is also weakly lower semicontinuous and from \eqref{saddlepoint} it follows that for $x\in V^{ac}_r$, 
\begin{align*}
L(x,\bar y )&\leq \liminf_j L(x,\hat y(\hat x^{\varepsilon_j})) \leq \liminf_j L_{\varepsilon_j}(\hat x^{\varepsilon_j},\hat y(\hat x^{\varepsilon_j}))  \\ &\leq \liminf_j\sup_{x\in V^{ac}_r} \inf_{y\in L^\infty([0,T]) }L_{\varepsilon_j}(x,y) \leq \sup_{x\in V^{ac}_r} \inf_{y\in L^\infty([0,T]) }L(x,y),
\end{align*}
which, together with the standard minimax inequality, proves that 
\begin{align*}
&\sup_{x\in V^{ac}_r} \inf_{y\in L^\infty([0,T]) }L(x,y) = \inf_{y\in L^\infty([0,T])} \sup_{x\in V^{ac}_r}  L(x,y)\\
& = \inf_{y\in L^\infty([0,T])} \sup_{x\in V^{ac}_r}  \{2H(x) - \int_{[0,T]} \dot x_t\cdot  y_t  dt +\int_{[0,T]} G(y_t) dt \}.
\end{align*}

Step 3. Given $y\in L^\infty([0,T])$, we can find a sequence of functions $(y^n)_{n\geq 1}$ belonging to $V_r$ with $\|y^n\|_{\infty}\leq \|y\|_\infty$, which converges to $y$ in Lebesgue measure on $[0,T]$ (use Lusin's theorem plus an approximation of continuous functions with functions of bounded variation). Then, by the dominated convergence theorem, as $n\to 0$, 
$$
\int_{[0,T]} \dot x_t \cdot y^n_t dt  \to \int_{[0,T]} \dot x_t \cdot y_t dt.
$$ 
On the other hand, letting $A_n = \{x\in \mathbb R^n: G(x)\leq n\}$, we have that, for each $m\geq 1$, 
$$
\int_{[0,T]} G(y^n_t\mathbf 1_{y^n_t \in A_m}) dt \to \int_{[0,T]} G(y_t\mathbf 1_{y^n_t \in A_m}) dt
$$
by dominated convergence and, for each $n\geq 1$, 
$$
\int_{[0,T]} G(y^n_t\mathbf 1_{y^n_t \in A_m}) dt \to \int_{[0,T]} G(y^n_t) dt 
$$
by monotone convergence. These two observations together with an integration by parts for the second term, imply that
$$
\inf_{y\in L^\infty([0,T])} \sup_{x\in V^{ac}_r}  L(x,y) = \inf_{\mu \in M} \sup_{x\in V^{ac}_r}  \{2H(x) - \int_{[0,T]}x_t\cdot d\mu +\int_{[0,T]} G(\mu([t,T])) dt \}.
$$
In addition, the assumption of the proposition implies that the inner supremum may also be computed over $x\in V_r$. 

Step 4. Finally, the following computation allows to finish the proof. 
\begin{align*}
&\inf_{\theta\in M} \sup_{x\in V_r} \{2H(x) - \int_{[0,T]}  x_t \cdot \theta(dt) + \int_{[0,T]} G(\theta([t,T]))dt  - \bar J(x)\}\\ &\qquad=  \inf_{\theta\in M} {\sup_{x\in V_r}\inf_{\mu\in M}}\{2H(x) - \int_{[0,T]} x_t
(\theta(dt)+ \mu(dt)) \\ &\qquad\qquad \qquad +\int_{[0,T]} G(\theta([t,T]))dt + \int_{[0,T]}
  G(\mu([t,T])) dt\}\\
& \qquad= \inf_{\theta\in M} {\inf_{\mu\in M}\sup_{x\in V_r}}\{2H(x) - \int_{[0,T]} x_t
(\theta(dt)+ \mu(dt)) \\ &\qquad\qquad \qquad +\int_{[0,T]} G(\theta([t,T]))dt + \int_{[0,T]}
  G(\mu([t,T])) dt\} \\
& \qquad = \inf_{\theta\in M} \inf_{\mu\in M}\{2\widehat H\left(\frac{\theta+\mu}{2}\right) +\int_{[0,T]} G(\theta([t,T]))dt + \int_{[0,T]}
  G(\mu([t,T])) dt\} \\
& \qquad=2\inf_{\theta \in M}\{\widehat H(\theta) +
\int_{[0,T]} G(\theta([t,T]))dt\},
\end{align*}
where the last equality follows by convexity of $G$, taking $\mu = \theta$. 

\end{proof}

\paragraph{Concavity of the log-payoff} The concavity of the log-payoff function $H(x)$ may be tested using the following simple lemma. We recall that $X_i = \log S_i$ for $i=1,\dots,n$. 

\begin{lemma}\label{conc.lm}
Let $\tilde P(X) = P(S)$ and assume that $\tilde P$ is concave on the set $\mathcal X^+:=\{x \in D: \tilde
  P(x)>0\}$ and that the set $\mathcal X^+$ is convex. Then the log-payoff $H: D \mapsto \overline {\mathbb R}$ is concave in $x$.
\end{lemma}
\begin{proof}
Let $0<\alpha<1$ and choose $a,b\in \mathcal X^+$. Then,
\begin{align*}&\alpha H(a) + (1-\alpha)H(b) = \alpha \log \tilde P(a) +
(1-\alpha)\log \tilde P(b) \\ &\leq \log
(\alpha \tilde P(a) + (1-\alpha)\tilde P(b))\leq \log \tilde P(\alpha
a + (1-\alpha) b) = H(\alpha a + (1-\alpha) b),
\end{align*}
which shows that $H$ is concave on $\mathcal X^+$. Since $H(x) = -\infty$ for $x\notin \mathcal X^+$ and the set $\mathcal X^+$ is convex, $H$ is also concave on the whole space. 
\end{proof}

\section{Examples}\label{ex.sec}
In this section, we specialize the results of the previous section to several option pay-offs encountered in practice. Throughout this section we assume that the L\'evy process satisfies the assumptions (A1) and (A2). For each considered pay-off, we need to check the assumptions of Proposition \ref{varadhanapplies} to ensure that the asymptotically optimal variance reduction measure $\theta^*$ may indeed be defined as in Definition \ref{optvar.def}, and the assumptions of Theorem \ref{maindual.thm}, to ensure that one can use the simplified formula \eqref{dualtheta} to compute $\theta^*$.

\paragraph{General European pay-off} We first check the assumptions of Theorem \ref{maindual.thm} and show that the problem of finding the optimal parameter $\theta^*$ becomes finite-dimensional. The assumptions of Proposition \ref{varadhanapplies} can be checked on a case-by-case basis as will be illustrated below. 
\begin{proposition}\label{euro.prop}
Assume that $H\left((x_t)_{0\leq t\leq T}\right) = h(x_T)$ with
$h:\mathbb R^n \to \mathbb R$ concave and upper semicontinuous. Then, assumptions of Theorem \ref{maindual.thm} are satisfied and 
$\theta^* = \bar\theta^* \delta_T$, where $\delta_T$ is the Dirac measure at $T$, and 
$$
\bar\theta^* = \arg\min_{\theta\in \mathbb R^n} \{\hat h(\theta) + T G(\theta)\},
$$
where $\hat h(\theta) = \sup_{v\in \mathbb R^n}\{h(v) -  v\theta\}$.
\end{proposition}
\begin{proof}
The log-payoff $H$ clearly satisfies the assumptions of Theorem \ref{maindual.thm}. If $\theta([0,T))\neq 0$, then  $\widehat H(\theta)=+\infty$ since one can choose $x_t = a\mathbf 1_{t<T}$ with $a$ arbitrary. This means that one can restrict the optimization in \eqref{dualtheta} to measures of the form $\theta \delta_T$ where $\theta \in \mathbb R^n$, and the rest of the proof follows easily. 
\end{proof}
\begin{remark}
We observe that the function $G(\theta)$ is known explicitly in most models. In addition, under the measure $\mathbb P^\theta$, $X$ is still a L\'evy
  process which often  falls into the same parametric class (see e.g., the variance gamma example in the following section). 
Thus, the only overhead of using the importance sampling estimator proposed in this paper for European options is due to the additional time needed to solve an explicit convex optimization problem in dimension $n$, which is usually negligible. 

\end{remark}
\paragraph{European basket put option} Now consider a specific European pay-off of the form $P(S_1,\dots,S_n) = (K-S_1-\dots-S_n)^+$. Then, using the notation of the previous paragraph,
$$
h(x_1,\dots,x_n) = \log(K - e^{x_1}-\dots-e^{x_n})^+.
$$
Since this functions is bounded from above and continuous on the set where it is not equal to $-\infty$, assumptions of Proposition \ref{varadhanapplies} are satisfied and one can define the asymptotically optimal variance reduction measure.
On the other hand, the function $\tilde P = (K - e^{x_1}-\dots - e^{x_n})^+$ is concave on
$\{\tilde P>0\}$ by convexity of the exponential and the set 
$\{e^{x_1} + \dots e^{x_n}<K\}$ is convex. Therefore, by Lemma \ref{conc.lm}, the function $h$ is concave. Since it is also clearly upper semicontinuous, the optimal measure $\theta^*$ is given in Proposition \ref{euro.prop}, where the convex conjugate of $h$ is easily shown to be explicit and given by 
$$
 \hat h(\theta) = \left\{\begin{aligned}
&+\infty && \theta_k \geq 0 \text{ for some $k$}\\
& -\left(1-\sum_k \theta_k\right)\log \frac{1- \sum_k \theta_k}{K} - \sum_{k} \theta_k \log(-\theta_k) &&\text{otherwise.}
\end{aligned}\right.
$$
Numerical examples for the European basket put option are given in the next section. 

\paragraph{Asian put option} In this example we consider the Asian option with log-payoff
$$H(x) = \log \left(K - \frac{1}{T}\int_0^Te^{x_t } dt
\right)^+.
$$ 
First note that $H$ may not be continuous in the $\sigma(D,M)$-topology even on the set where it is finite. We shall nevertheless use the definition \eqref{optvar.def} of the asymptotically optimal variance reduction parameter. This may be justified by the fact the discretely sampled Asian option is $\sigma(D,M)$-continuous, and the variance of the discretely sampled Asian pay-off converges to that of the continuously sampled pay-off as the discretization step tends to zero. 

Let us now check the assumptions of Theorem \ref{maindual.thm}. 
Remark that $K - \frac{1}{T}\int_0^Te^{x_t} dt$ is concave by
  convexity of the exponential, and for $x,y\in D$ such that $\frac{1}{T}\int_0^Te^{x_t} dt <
  K$ and $\frac{1}{T}\int_0^Te^{y_t} dt< K$, 
$$
\frac{1}{T}\int_0^Te^{\alpha x_t + (1-\alpha) y_t } dt
\leq \frac{1}{T}\int_0^T\left(\alpha e^{x_t} + (1-\alpha)
e^{y_t}\right) dt< K,
$$
which implies that the set $\{\tilde P>0\}$ is convex. Therefore, by Lemma \ref{conc.lm}, $H$ is concave. Moreover, assume that $x^n\to x$ in $V^{ac}_r$. Then, by Fatou's lemma
$$
\lim\inf_n \int_0^Te^{x^n_t} dt \geq \int_0^Te^{x_t} dt,
$$
which shows that $H$ is upper semicontinuous. Finally, $x\in V_r$ may be approximated by a uniformly bounded sequence of $(x_n)\subset V^{ac}_r$, so that $H(x_n)\to H(x)$ by the dominated convergence theorem. Therefore, all assumptions of Theorem \ref{maindual.thm} are satisfied by the Asian put option. 
The convex conjugate of $H$ and the asymptotically optimal parameter $\theta^*$ are described by the following proposition.
\begin{proposition}${}$
\begin{itemize}
\item[i.] If $\theta$ is absolutely continuous, with density (also denoted
  by $\theta_t$) satisfying $\theta_t\leq 0$ for all $t\in [0,T]$, then $\widehat H(\theta)$ is given by
$$
\widehat H(\theta) = \log \frac{K}{1 - \int_0^T \theta_t dt}  -
\int_0^T \theta_t \log \frac{-K T \theta_t}{1-\int_0^T \theta_s ds} dt.
$$
Otherwise $\widehat H(\theta) = +\infty$. 
\item[ii.] The function $\psi^*_t = \int_{T-t}^T\theta^*_s ds$ is the
solution of the boundary value problem
\begin{align*}
&\dot p_t = -G'(\psi^*_t),&& p_T = -\log\frac{K}{1-\psi^*_T} +1,\\
& \dot \psi^*_t = -\frac{1}{T}e^{p_t+1},&&\psi^*_0 = 0.
\end{align*}
\end{itemize}
\end{proposition}
\begin{remark}
This system can be integrated explicitly:
$$
\ddot\psi_t = -\dot\psi_t G'(\psi_t) \quad \Rightarrow\quad \dot\psi_t = -G(\psi_t) - C\quad \Rightarrow\quad \int_0^{\psi_t}\frac{d\phi}{C + G(\phi)} = -t,
$$
where the constant $C$ is determined from the terminal condition. 

\end{remark}
\begin{proof}
By definition,
$$
\widehat H(\theta) = \sup_{x\in V_r}\left\{\log \left(K - \frac{1}{T}\int_0^Te^{x_t } dt
\right)^+ - \int_{[0,T]} x_t \theta(dt)\right\}.
$$
First, assume that there is an interval $[a,b)\in [0,T]$ such that $\theta([a,b))>0$. Then, letting $x_t = \log K$ for $t\notin [a,b)$ and $x_t = -N$ for $t\in[a,b)$, and making $N$ tend to $+\infty$, we see that $\widehat H(\theta)=+\infty$. Therefore, from now on we may assume that $\theta$ is a negative measure. Assume that it is not absolutely continuous. Then, there exists $\varepsilon>0$ such that for all $\delta >0$, there exists a finite sequence of pairwise disjoint sub-intervals $[x_k,y_k)$ of $[0,T]$ satisfying 
$$
\sum_k (y_k-x_k)<\delta
$$
such that 
$$
\sum_k \theta([x_k,y_k))<-\varepsilon.
$$
We define $I = \cup_{k}[x_k,y_k)$. 
Let $x_t = N> \log (K/2)$ when $t\in I$ and $x_t = \log (K/2)$ otherwise. Then,
$$
\log \left(K - \frac{1}{T}\int_0^Te^{x_t } dt
\right)^+ - \int_{[0,T]} x_t \theta(dt)\geq \log \left(K/2 - \frac{\delta}{T}e^N
\right)^+ - \theta(I^c)\log \frac{K}{2}  +N\varepsilon.
$$ 
Taking $N$ sufficiently large and $\delta$ sufficiently small, we see that $\widehat H(\theta) = +\infty$ in this case as well. We may therefore assume that $\theta$ is an absolutely continuous negative measure, and, with an abuse of notation, its density will also be denoted by $\theta_t$. The computation of $\widehat H(\theta)$ reduces to computing the supremum 
\begin{align}
\sup_{x}\left\{\log \left(K - \frac{1}{T}\int_0^Te^{x_t } dt
\right) - \int_{[0,T]} x_t \theta_tdt\right\},\label{asianHhat}
\end{align}
where with no loss of generality we may consider only those $x\in V_r$ for which the expression under the sign of logarithm is positive. The first order condition for this optimization problem writes
$$
\theta_t = -\frac{1}{T}\frac{ e^{x_t}}{K - \frac{1}{T}\int_0^Te^{x_s } ds}. 
$$
Integrating this expression from $0$ to $T$, we find
$$
K-\frac{1}{T}\int_0^Te^{x_s } ds  =\frac{ K }{1- \int_0^T \theta_t dt}
$$
and so 
$$
x_t = \log \frac{ - K T \theta_t }{1- \int_0^T \theta_t dt}.
$$
substituting this into \eqref{asianHhat}, we obtain the first part of the proposition.

To compute the optimal importance sampling measure $\theta^*$, we need to 
solve
\begin{align*}
&\min_{\theta\in M} \{\widehat H(\theta) + \int_0^T G(\theta([t,T]))
dt\}\notag\\
& = \min_{\theta_t \leq 0}\Big(1-\int_0^T \theta_s ds\Big)\log \frac{K}{1 - \int_0^T \theta_t dt}  -
\int_0^T \theta_t \log(-T \theta_t) dt +  \int_0^T
G\Big(\int_t^T\theta_s\Big)dt
\end{align*}
Introducing the function $\psi_t = \int_{T-t}^T\theta_s ds$, the optimization problem becomes
$$
\min_{\dot\psi_t \leq 0}\Big(1-\psi_T\Big)\log \frac{K}{1 - \psi_T}  -
\int_0^T \dot\psi_t \log(-T \dot\psi_t) dt +  \int_0^T
G(\psi_t)dt.
$$
By Pontryagin's maximum principle for deterministic control problems, for every $t\in [0,T]$, the optimal $\dot\psi_t$ is the minimizer of 
$$
p_t\dot \psi_t - \dot\psi_t \log(-T\dot \psi_t),
$$
where $p_t$ is the ``adjoint state'' satisfying
$$
\dot p_t = - G'(\psi^*_t),\qquad p_T = 1 - \log K + \log(1-\psi^*_T). 
$$
\end{proof}
Numerical examples for the Asian put option are given in the next section. 

\section{Numerical illustrations}\label{num.sec}
In this section, we illustrate the results of this paper with numerical computations in the multivariate variance gamma model. Let $b\in \mathbb R^n$, $\Sigma$ be a positive definite $n\times n$
matrix, and define
$$
X_t = \mu t + b \Gamma_t + \Sigma W_{\Gamma_t},
$$
where $W$ is a standard Brownian motion in dimension $n$, $\Gamma$ is a gamma process with $\mathbb E[\Gamma_t] = t$ and
$\text{Var}\, \Gamma_t = t/\lambda$, and $\mu$ is chosen to have
$\mathbb E[e^{X^i_t}] = 1$ for all $t$ and $i=1,\dots,n$. Then, the cumulant generating function $X_1$ under the original measure is given by 
$$
G(\theta) =  \langle\theta,\mu\rangle - \lambda \log \left(1- \frac{\langle\theta,
    b\rangle}{\lambda} - \frac{\langle\Sigma \theta, \theta
    \rangle}{2\lambda}\right), \theta \in \mathbb R^n. 
$$
with 
$$
\mu^i = \lambda \log \left(1- \frac{
    b^i}{\lambda} - \frac{ \Sigma_{ii}}{2\lambda}\right),\quad
i=1,\dots,n. 
$$
Under the measure $\mathbb P^\theta$, the cumulant generating function of $X_1$ can be written as
$$
G^\theta(u) = \langle u,\mu\rangle - \lambda \log \left(1-
  \frac{\langle u,
    b + \Sigma\theta\rangle}{\lambda u^*} - \frac{\langle\Sigma u,  u
    \rangle }{2\lambda u^*}\right) ,\quad u^* = 1- \frac{\langle\theta,
    b\rangle}{\lambda} - \frac{\langle\Sigma \theta, \theta
    \rangle}{2\lambda}.
$$
Therefore, under the measure $\mathbb P^\theta$, the process $X^1$ is also a variance gamma process with parameters $\tilde \mu = \mu$, $\tilde \lambda = \lambda$, $\tilde b = \frac{b+\Sigma\theta}{u^*}$ and $\tilde \Sigma = \frac{\Sigma}{u^*}$.

\paragraph{Vanilla put in the variance gamma model}

In the first example, we let $n=1$ and price a European put option with pay-off
$P(S) = (K-S)^+$. The model parameters are $\lambda = 1$, $b = -0.2$ and $\sqrt{\Sigma} =
0.2$, which corresponds to annualized volatility of $28\%$, skewness
of $-1.77$ and excess kurtosis of $2.25$. Table \ref{vred1d.tab} shows the variance reduction ratios and the values of the asymptotically optimal parameter $\theta^*$ as function of strike and time to maturity. We see that the highest ratios are attained for out-of-the-money options, whose exercise is a rare event, but that even for at-the-money options, the variance reduction ratios remain quite significant. It is also important to understand, how close are these ratios to the optimal ones which would have been obtained by minimizing the actual variance of the estimator rather than its asymptotic proxy. This is illustrated in Figure \ref{vred1d.fig}, which plots the variance of the importance sampling estimator (evaluated by Monte Carlo) as function of the parameter $\theta$. We see that for the chosen parameter values $\theta^*$ is very close to optimality.

\begin{table}
\begin{center}
\begin{tabular}{l|cccccc}
$T$ & $0.25$ & $0.5$ & $1$ & $2$ & $3$ \\\hline
Optimal parameter $\theta^*$ & $-2.77$ & $-2.45$ & $-2.06$& $-1.65$ & $-1.41$\\
Variance ratio &$3.38$ &  $3.61$&  $3.78$  &  $3.75$ &  $3.67$\\\hline \hline
$K$ & $0.5$ & $0.7$ & $0.9$ & $1.1$ & $1.3$ & $1.5$\\\hline
Optimal parameter $\theta^*$ &$-2.84 $&$-2.56$&$ -2.24$&$ -1.88$&$ -1.54$&$ -1.25$\\
Variance ratio &$17.44$ &  $6.80$ &  $4.14$ &  $3.19$&
         $3.63$ &  $3.63$
\end{tabular}
\end{center}
\caption{Single-asset European put option. Top: Variance reduction ratios as function of time to maturity $T$, for
$K=1$.
Bottom: Variance reduction ratios as function of strike $K$, for
$T=1$.}
\label{vred1d.tab}
\end{table}
\begin{figure}
\centerline{\includegraphics[width=0.55\textwidth]{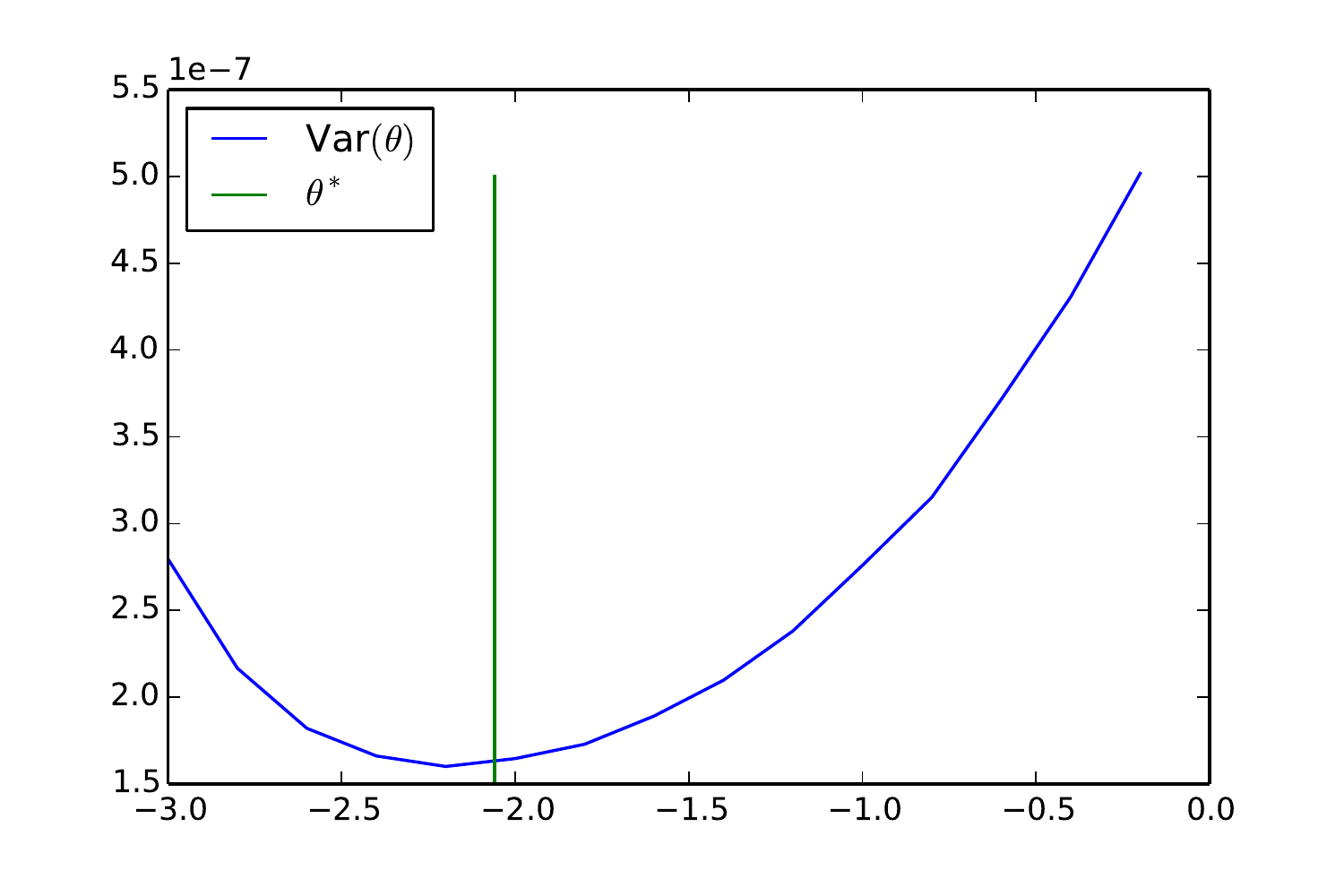}\hspace*{-0.5cm}\includegraphics[width=0.55\textwidth]{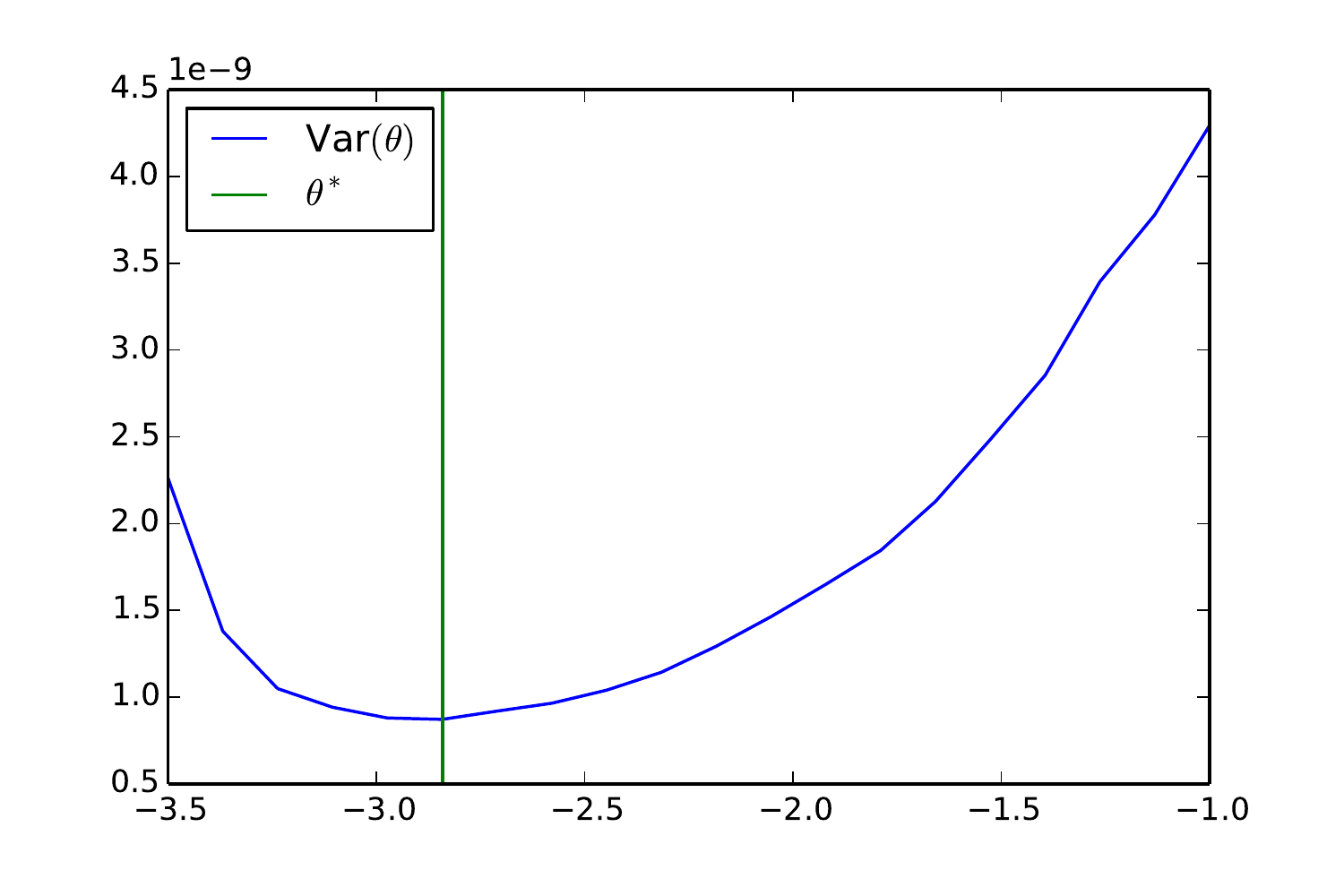}}
\caption{Optimality of $\theta^*$. Left: $T=1$, $K=1$. Right: $T=1$, $K=0.5$. }
\label{vred1d.fig}
\end{figure}

\paragraph{Basket put in the variance gamma model}
In this example, we let $n=3$ and price a European basket put option with pay-off
$P(S) = (K-S^1 - S^2-S^3)^+$. The model parameters are
$$
\lambda = 1,\quad b = \left(\begin{aligned}-0.2 \\ -0.2 \\ -0.2\end{aligned}\right)\quad \text{and}\quad \Sigma = \left(\begin{aligned}
&0.04  &&0.02 &&0.02 \\
&0.02  &&0.04 &&0.02 \\
&0.02  &&0.02 &&0.04 
\end{aligned}\right).
$$
Table \ref{vred3d.tab} shows the variance reduction ratios and the values of the asymptotically optimal parameter $\theta^*$ as function of strike and time to maturity. We see that the values are similar to the one-dimensional case.

\begin{table}

\begin{center}
\begin{tabular}{l|ccccccc}
$T$ & $0.25$ & $0.5$ & $1$ & $2$ & $3$ \\\hline
Variance ratio &$3.55$ &  $3.67$&  $3.85$  &  $3.81$ &  $3.76$\\\hline\hline
$K$ & $1.5$ & $2$ & $2.5$ & $3$ & $3.5$ & $4$ & $4.5$ \\
Variance ratio, $T=1$ & $23.1$ & $9.78$ & $5.53$ & $3.80$ & $3.23$ &
$4.22$ & $5.14$\\
Variance ratio, $T=3$ & $6.63$ & $4.88$ & $4.35$ & $3.81$ & $2.96$ &
$2.42$ & $2.19$
\end{tabular}
\end{center}
\caption{European basket put option. Top: Variance reduction ratios as function of time to maturity $T$, for
$K=1$. Bottom: variance reduction ratios as function of strike for $T=1$ and $T=3$.}
\label{vred3d.tab}
\end{table}

\paragraph{Asian put in the variance gamma model}
In this final example we price an Asian put
option with pay-off $P(S) = \left(K- \frac{1}{T}\int_0^T S_t
  dt\right)^+$, for $T=1$. Table \ref{vredasian.tab} shows the variance reduction ratios as function of strike, and Figure \ref{vredasian.fig} plots the ``distribution function'' $\theta^*([t,T])$ of the asymptotically optimal measure $\theta^*$ as function of time $t$ for $K=1$. We see that the variance reduction ratios are even better than the ones obtained for the European call, since for low strikes the exercise probability is smaller for the Asian option than for the European option with the same strike and maturity. 

\begin{table}

\begin{center}
\begin{tabular}{l|cccccc}
$K$ & $0.5$ & $0.7$ & $0.9$ & $1.1$ & $1.3$ & $1.5$\\\hline
Variance ratio &$39.7$ &  $10.6$ &  $4.82$ &  $3.21$&
         $5.08$ &  $6.91$
\end{tabular}
\end{center}
\caption{Asian put option. Variance reduction ratios as function of strike for $T=1$.}
\label{vredasian.tab}
\end{table}

\begin{figure}
\centerline{\includegraphics[width=0.6\textwidth]{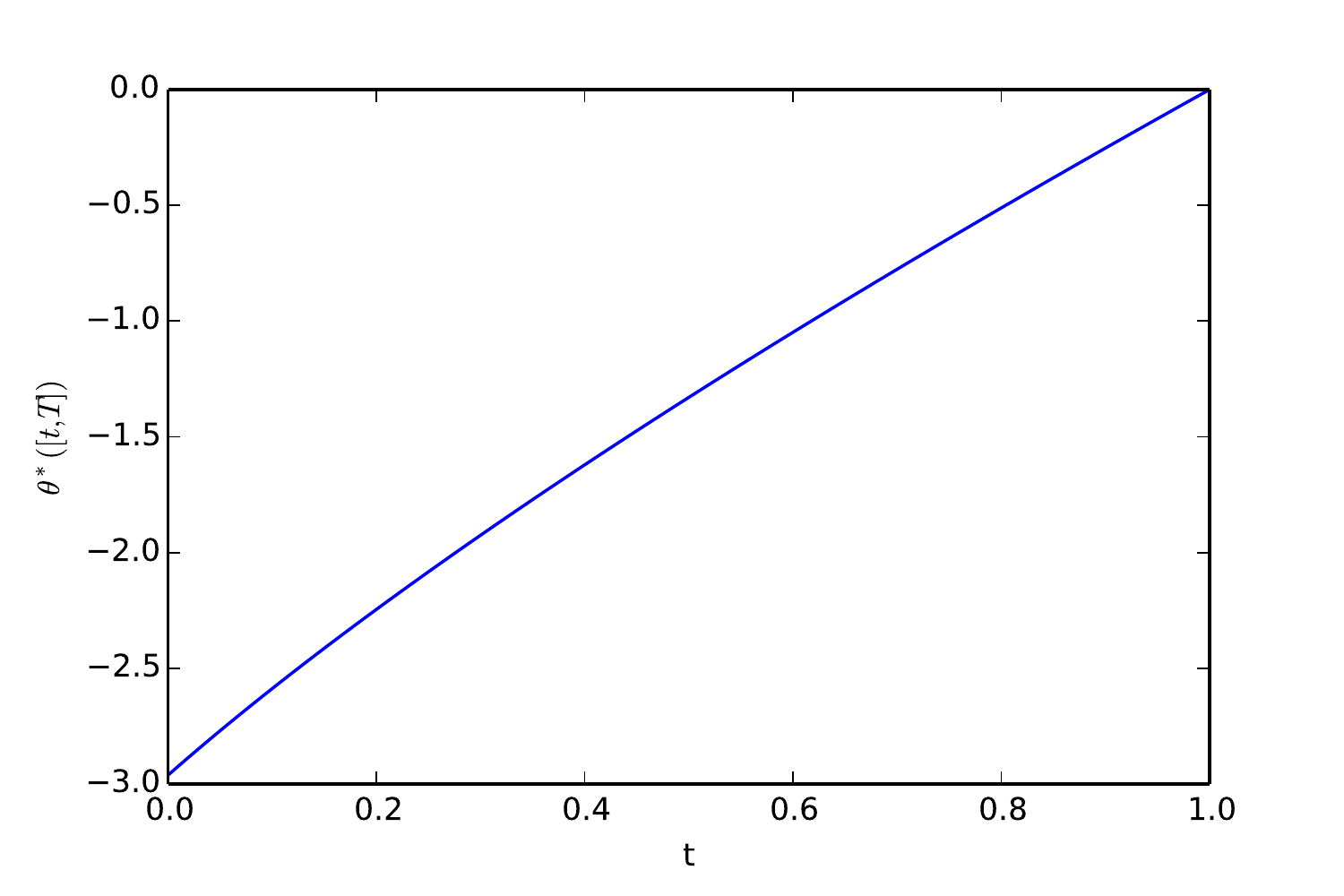}}
\caption{Asian put option. Distribution function of the asymptotically optimal measure $\theta^*$, for $K=1$ and $T=1$.}
\label{vredasian.fig}
\end{figure}

\section*{Appendix}
\begin{lemma}\label{wienerhopf.lm}
Let $X$ be a L\'evy process. We denote $\overline X_t  = \sup_{0\leq s\leq t}
X_s$. Let $\beta'>0$ be such that $\mathbb
E[e^{\beta' X_t }]<\infty$ and $\mathbb E[e^{\beta' X_t} X_t]
< 0$. Then,
$$
\mathbb E[e^{\beta \overline X_\infty}] < \infty
$$
for all $\beta \in (0,\beta')$. 
\end{lemma}

\begin{proof}
By the Wiener-Hopf factorization (Theorem 6.16 in \cite{Kyprianou},
see also Equation (47.9) in \cite{sato}),
$$
q\int_0^\infty e^{-qt} \mathbb E[e^{\beta \overline X_t}] dt =
\exp\left(\int_0^\infty \frac{e^{-qt}}{t} dt
    \int_0^\infty (e^{\beta x} - 1)  \mathbb P[X_t\in
    dx] \right).
$$
Letting $q\to 0$ and using the Cauchy-Schwarz inequality, we get
\begin{align*}
\mathbb E[e^{\beta \overline X_\infty}] &= \exp\left(\int_0^\infty \frac{1}{t} dt
    \int_0^\infty (e^{\beta x} - 1)  \mathbb P[X_t\in
    dx] \right)\\
& = \exp\left(\int_0^\infty \frac{1}{t} dt\, 
    \mathbb E[(e^{\beta X_t}-1)\mathbf 1_{X_t \geq 0}] \right)\\
& \leq \exp\left(\int_0^\infty \frac{\beta}{t} dt\, 
    \mathbb E[X_t\,e^{\beta X_t}\mathbf 1_{X_t \geq 0}] \right)\\
& \leq \exp\left(\int_0^\infty \frac{\beta}{t} dt\, \mathbb
  E[|X_t|^{\frac{\beta'}{\beta'-\beta}}\mathbf 1_{X_t\geq 0}]^{1-\frac{\beta}{\beta'}}\,
    \mathbb E[e^{\beta' X_t}\mathbf 1_{X_t \geq 0}]^{\frac{\beta}{\beta'}} \right)
\end{align*}

Let
$\widehat X_t = X_t - \sum_{s\leq t:\Delta X_s < -1} \Delta
X_s$. Then, all moments of $\widehat X_t$ are finite and 
$$
\mathbb
  E[|X_t|^{\frac{\beta'}{\beta'-\beta}}\mathbf 1_{X_t\geq 0}] \leq
  \mathbb E[|\widehat X_t|^{\frac{\beta'}{\beta'-\beta}}]. 
$$
Let $n^* = \inf\{n\in \mathbb N: n\geq
\frac{\beta'}{\beta'-\beta}\}$. Then, 
$$
\mathbb E[|\widehat X_t|^{\frac{\beta'}{\beta'-\beta}}] \leq \mathbb
E[|\widehat X_t|^{n^*}]^{^{\frac{\beta'}{n^*(\beta'-\beta)}}} \leq C t^{^{\frac{\beta'}{n^*(\beta'-\beta)}}}
$$
for some $C<\infty$. Thus,
\begin{align}
\mathbb E[e^{\beta \overline X_\infty}] \leq \exp\left(\beta
  C\int_0^\infty dt\, t^{\frac{1}{n^*}-1}\,
    \mathbb E[e^{\beta' X_t}\mathbf 1_{X_t \geq 0}]^{\frac{\beta}{\beta'}} \right)\label{int}.
\end{align}

Further, 
$$
\mathbb E[e^{\beta' Y_t}\mathbf 1_{Y_t \geq 0}] = e^{t\psi(\beta)}
\widehat{\mathbb P}[X_t\geq 0],
$$
where $\frac{d\widehat {\mathbb Q}}{d\mathbb Q}|_{\mathcal F_t} =
\frac{e^{\beta' X_t}}{\mathbb E[e^{\beta' X_t}]}. $
By Cramer's theorem, 
$$
\lim_{t\to \infty}\frac{1}{t}\log \widehat{\mathbb P}[X_t\geq 0] = -\inf_{x\geq
  0}\Lambda^*(x) = - \Lambda^*(0),
$$
where 
$$
\Lambda^*(x) = \sup_{\lambda} \{\lambda x - \log \widehat{\mathbb
  E}[e^{\lambda X_1}]\} = \sup_{\lambda} \{\lambda x - \psi(\lambda +
\beta) +\psi(\beta)\}. 
$$
Therefore,
$$
\lim_{t\to \infty} \frac{1}{t}\log \left(e^{t\psi(\beta)}
\widehat{\mathbb P}[X_t\geq 0] \right) = \inf_{\lambda} \psi(\lambda +
\beta) < 0,
$$
which shows that the integral in \eqref{int}
converges. 
\end{proof}

\end{document}